\documentclass[runningheads]{llncs}

\usepackage[T1]{fontenc}
\usepackage{graphicx}
\usepackage{cite}
\usepackage{algorithm,algorithmic}
\usepackage{xspace}
\usepackage{svg}
\usepackage{caption}
\usepackage{booktabs}
\usepackage{amsmath}
\usepackage{amsfonts}
\usepackage{wrapfig}
\usepackage{tikz}
\usepackage{hyperref}
\usetikzlibrary{arrows.meta}

\newcommand{\hb}{\emph{heartbeat}}
\newcommand{\bm}{\emph{Band Manager}}
\newcommand{\issA}{\emph{Type-A Issuer}}
\newcommand{\issB}{\emph{Type-B Issuer}}
\newcommand{\sysname}{\emph{SpexPay}\xspace}

\newcommand{\algcomment}[1]{\hfill{\scriptsize{// \textit{#1}}}}

\begin{document}
\title{\sysname: A Privacy-Preserving Pay-As-You-Go System for Dynamic Spectrum Sharing}
%
%

\author{Mohaimin Al Barat\inst{1} \and
Hexuan Yu\inst{1} \and
Shaoyu Li\inst{1} \and 
Yang Xiao\inst{2} \and
Yi Shi\inst{1} \and \\
Eric W. Burger\inst{1} \and 
Y. Thomas Hou\inst{1} \and
Wenjing Lou\inst{1}
}

\titlerunning{\sysname: A Privacy-Preserving Pay-As-You-Go System}
\authorrunning{M.A. Barat et al.}
%
\institute{Virginia Tech, USA \and
University of Kentucky, USA
}
\maketitle              
\begin{abstract}
Dynamic Spectrum Sharing (DSS) is a cornerstone of next-generation wireless systems, yet existing solutions such as Spectrum Access Systems (SAS) rely on centralized administrators that expose sensitive operational metadata and lack cryptographic transaction accountability. Though SAS administrators, such as Google, have introduced pay-as-you-go pricing models~\cite{googleSpectrumAccess}, these approaches still face significant privacy and accountability challenges as DSS evolves toward a more open and large-scale spectrum marketplace. We present \sysname, a privacy-preserving and auditable pay-as-you-go spectrum usage framework that enforces fine-grained, usage-linked payments without revealing user identities. \sysname~integrates BBS$+$ verifiable credentials, unlinkable session pseudonyms, and selective-disclosure proofs to enforce privacy-preserving access authorization, while leveraging Solidity-based smart contracts to realize automated and non-repudiable escrow settlement. By recording only pseudonymous usage evidence and hash-chained metering data on-chain, the system achieves strong unlinkability, preserving verifiable accountability and auditability. A full prototype demonstrates low end-to-end latency ($\approx$150 ms) and modest on-chain cost ($\approx$603K gas or $\approx$\$0.9), showing that \sysname~is practical for real-world DSS deployments. We also evaluated the user-side cryptographic operations on a Raspberry Pi 5 to assess scalability and suitability for {edge-class hardware}. Our code and artifacts are publicly available at \url{https://github.com/iambarat/SpexPay}.


\end{abstract}


\section{Introduction}
\label{sec:intro}

The radio spectrum is the fundamental enabler of modern wireless communication and sensing systems. 
To accommodate the explosive growth of 5G, Wi-Fi~6/7, satellite, and radar applications, regulators such as the U.S. Federal Communications Commission (FCC) and the National Telecommunications and Information Administration (NTIA) have opened portions of previously exclusive federal bands for shared civilian use~\cite{fcccbrs,fcc_cband_transition,fcc_3450_service}, giving rise to dynamic spectrum sharing (DSS). In principle, DSS allows multiple users to access common bands under a sharing framework, alleviating spectrum scarcity by increasing the utilization of existing bands. The practical relevance of this setting is grounded in the Citizens Broadband Radio Service (CBRS), where access to the 3.5\,GHz band is coordinated by a Spectrum Access System (SAS) under the FCC Part~96 framework. The FCC has approved multiple SAS administrators for commercial deployment, including CommScope, Federated Wireless, Google, and Sony~\cite{fcc_part96_cbrs,fcccbrs}. Commercial DSS management frameworks further allow registered users to request spectrum access based on the pay-as-you-go (PAYG) model, such as Google's pricing~\cite{googleSpectrumAccess}, which subjects all users to a market mechanism for spectrum access rather than granting them unconditional preemptive access as was done in the past.

As demand outpaces licensed holdings, PAYG is increasingly seen as critical for market-driven access to mid- and high-band spectrum. Unlike long-term licenses, it dynamically reallocates spectrum to its highest-valued use and reduces underutilization. By aligning consumption with economic incentives, PAYG-based DSS could become the dominant model for next-generation shared bands, helping operators extract greater value from scarce spectrum while accelerating innovation in cellular, IoT, and sensing.

\textbf{Privacy Challenge in PAYG.} While PAYG-style DSS improves utilization, it changes spectrum access: rather than autonomous, exclusive use, users must transact with third-party \bm s (e.g., SAS administrators) for short-term leases. This raises privacy challenges: PAYG requires users to register, authenticate, and submit access requests, exposing sensitive information---identity, location, time, transmission power, bandwidth, and mission profiles---to commercial band managers.
{As an example, we can consider a Navy shipborne radar that obtains short-term shared access through a commercial PAYG path. Under today’s centralized SAS, each grant request discloses location, time, frequency, bandwidth, and dwell to a commercial administrator.} Aggregated over time, these transactions enable inference and linkability attacks, allowing \bm s or adversaries to reconstruct mobility trajectories, operational rhythms, or mission profiles. The commercial nature of \bm s further risks monetization or unauthorized exposure of such sensitive data.

Since PAYG envisions participation by all users, including public-sector ones such as military radars and radio astronomy stations~\cite{nsf_payg_solicitation_2024}, privacy is often not merely desired but legally mandated given their national importance~\cite{ntia-spectrum-strategy}. {Aside from the US, under Europe's Licensed Shared Access (LSA) framework~\cite{ecc-dec-1402,ec-dec-2014} and national local-licensing regimes~\cite{ofcom-sal}, the PAYG path exposes comparable metadata for military, PMSE, and private-network operators.} For all users, elevating privacy protection is crucial for scaling up the PAYG-based sharing economy. Hence we target the research question---\emph{can we design privacy-preserving mechanisms that fulfill the PAYG model's financial obligation of usage-based spectrum consumption without compromising a user's operational privacy?}

\textbf{Contributions.} We propose \sysname, a new privacy-preserving PAYG system for next-generation spectrum sharing economy. At a high level, our system enables fine-grained spectrum access control and supports micro-payments tied to actual utilization, enforced through cryptographic proofs of activity that preserve user anonymity and unlinkability. Our motivation stems from two observations:  
(1) Future DSS ecosystems will consist of diverse devices, requiring ephemeral access sessions rather than long-term leases.  
(2) Sustainable DSS markets demand transparent, auditable, yet privacy-preserving payment mechanisms to foster trust among regulators, service providers, and users.

Concretely, \sysname~employs a two-tiered cryptographic credential model over a public blockchain. The two tiers reflect the regulated DSS access framework: users obtain both a regulatory-eligibility credential and a commercial-authorization credential prior to accessing spectrum. To initiate a session, the commercial-type issuer deposits funds into an escrow smart contract and publishes session parameters on-chain. For each access request, the user generates a selectively disclosed verifiable presentation from both credentials to prove eligibility to the \bm. During the session, the device periodically emits precomputed hash-chained heartbeats, preventing replay or over-claiming; the \bm~claims payment by submitting the latest valid heartbeat with the session ID. The contract enforces correct settlement, bounds claim windows, and ensures fair billing. Commercial-type issuers maintain two registries for: (1) \emph{auditability} (private) - device IDs mapped by session IDs for the auditability of users' spectrum usage, and (2) \emph{revocation list} (shared with BMs) - a session revocation list to which band managers subscribe, enabling them to receive updates when a malicious user’s session is revoked.

While our design leverages general-purpose privacy primitives, DSS introduces domain-specific challenges that are not addressed in existing privacy-preserving payment or authentication systems. 
Spectrum access is usage-metered and continuous, demanding verifiable accounting rather than one-time transactions. Access must also be conditioned on regulatory eligibility, tightly coupling authentication with economic settlement. Finally, SAS deployments require auditability and dispute resolution under regulatory oversight without compromising user privacy. These requirements  fundamentally differ from traditional anonymous payment or credential systems, which do not support continuous usage proofs, eligibility-payment binding, or regulatory accountability.

In summary, this paper makes the following key contributions:
\begin{itemize}
    \item \textbf{PAYG Model for Spectrum Sharing:} We introduce \sysname, the first PAYG system for next-generation spectrum economies with fine-grained, usage-based micro-payments while supporting critical regulatory functions as in the incumbent paradigms.
    \item \textbf{Privacy-preserving spectrum access control and micro-payments:}  
    We design a dynamic access control scheme based on a two-tiered spectrum credential system, enabling user devices to authenticate for spectrum usage without revealing real identities to the \bm s. 
    \item \textbf{Prototype Evaluation:} We implement and evaluate the prototype, measuring $\approx$150 ms total computation time for credential issuance, proof generation, and verification, and $\approx$603K gas for smart contract settlement. Evaluation on a Raspberry Pi further confirms scalability to {edge-class hardware}.

\end{itemize}

\section{Related Work}
\label{sec:related_work}

\textbf{Privacy in Spectrum Sharing and Access System.}
Prior work on spectrum-sharing privacy has largely focused on mitigating inference attacks against incumbents and users. Bahrak et al.~\cite{bahrak2014protecting} and Clark et al.~\cite{clark2016can} propose obfuscation mechanisms to hide operational parameters in database-driven access systems. While effective against external observers, these assume a trusted spectrum manager and do not address leakage toward the manager itself---critical in fine-grained, usage-based settings. TrustSAS~\cite{grissa2019trustsas} identifies privacy risks in centralized SAS architectures and integrates blockchain with cryptographic primitives to improve transparency and trust. However, it relies on long-term registration and does not support dynamic usage verification or usage-linked payment enforcement. Pri-Share~\cite{yu2024pri} extends privacy to multi-SAS environments using MPC and threshold cryptography, enabling conflict-free allocation without revealing user requests, but does not incorporate usage-linked economic settlement. 

\textbf{Anonymous Credentials and Privacy-Preserving Payments.}
Coconut \cite{sonnino2019coconut} provides threshold-issued, re-randomizable credentials with unlinkable presentations, suitable for decentralized environments. Several cryptocurrency payment systems provide strong transaction privacy, but they are not designed for regulated, usage-metered services such as dynamic spectrum sharing. Zerocash~\cite{BenSasson2014Zerocash} enables fully shielded transactions using zk-SNARKs, while PrivPay~\cite{moreno2015privpay} achieves anonymous payments in decentralized credit networks. \emph{Monero} relies on ring signatures to obscure transaction linkability \cite{vanSaberhagen2013CryptoNote,Noether2014CryptonoteReview}. These systems are fundamentally general-purpose: they do not encode access eligibility or bind payments to externally verifiable service usage.

Existing spectrum-sharing systems improve allocation privacy but lack integrated economic enforcement, while anonymous credentials and private-payment systems do not support regulated, usage-metered access control. None of the prior work simultaneously achieves \emph{(i)} unlinkable access authentication, \emph{(ii)} verifiable usage metering, and \emph{(iii)} on-chain, non-repudiable settlement. Our system addresses this gap by tightly coupling anonymous credentials with cryptographic usage proofs and programmable payment enforcement for DSS.
\section{Preliminaries}
\label{sec:preliminaries}

\subsection{Dynamic Spectrum Sharing Fundamentals}

\begin{figure}
    \centering
    \includegraphics[width=.8\textwidth]{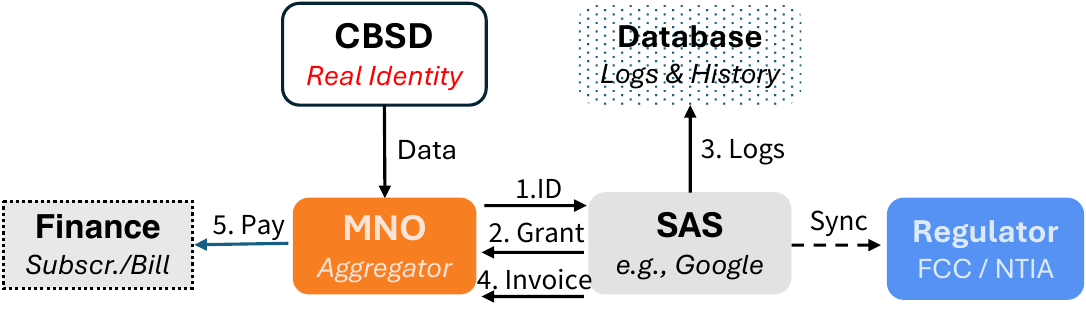}
    \caption{\small Current Centralized SAS Model.
    The Citizens Broadband Radio Service Device (CBSD) functions as the end-user spectrum device.
    The \textbf{MNO} (or Incumbents, e.g., naval user) acts as a domain proxy, forwarding the device's \textit{Real ID} and \textit{GPS} directly to the central \textbf{SAS} for access grants (1\&2), which logs sensitive operational data (3). Financial settlement is decoupled from actual usage, relying instead on rigid subscriptions or bulk invoices between the SAS and the MNO's \textbf{Financial System} (4\&5).}   
    \label{fig:sas_model}
    \vspace{8pt}
\end{figure}

DSS enables intensive reuse of underutilized bands by allowing multiple classes of users to coexist under regulatory coordination. In the U.S., the FCC regulation for the CBRS band (3.55--3.70~GHz) defines a tiered access model in which a SAS dynamically authorizes secondary users based on real-time spectrum availability while protecting the pre-emptive access rights of primary/incumbent users.
Similar sharing models exist for other bands: the AWS-3 band requires coordination between federal incumbents and commercial licensees to manage aggregate interference~\cite{fcc_aws3_report}; the 3.45--3.55~GHz band adopts a federal/non-federal sharing regime using cooperation planning areas and periodic use areas~\cite{fcc_3450_service}; and the C-Band transition illustrates incumbent relocation and repurposing for flexible commercial use~\cite{fcc_cband_transition}.


This shift from exclusive long-term licensing toward dynamic, short-term access calls for privacy-preserving authorization and fine-grained accounting. As shown in Fig.~\ref{fig:sas_model}, in systems such as CBRS, secondary users submit \textbf{grant requests} specifying operational parameters (frequency range, location, duration); upon \textbf{grant authorization}, the SAS admin issues time- and location-bound grants subject to interference protection and incumbent activity. To retain a grant, devices periodically transmit \textbf{heartbeats} signaling continued, compliant operation. As usage becomes increasingly transient, access is structured around short-lived, session-based grants rather than multi-year licenses, motivating tighter coupling of authorization, usage monitoring, and settlement.

\subsection{BBS$+$ Signatures \& Zero Knowledge Proof (ZKP)}
BBS$+$ signatures~\cite{au2006constant} extend the original Boneh–Boyen short signature scheme~\cite{boneh2004short} to support signing of multiple attributes and efficient selective disclosure. We briefly review BBS$+$ signatures following the convention from~\cite{camenisch2004signature,boneh2004short}. \\
A BBS$+$ signature scheme consists of three algorithms:
\begin{enumerate}
    \item $\mathsf{KeyGen}(1^\lambda)$: samples $x\leftarrow\mathbb{Z}_p$, outputs $(sk=x, pk=g_2^x)$.
    \item $\mathsf{Sign}(sk, (m_1,\dots,m_\ell))$: (a) sample $e,s\leftarrow\mathbb{Z}_p$, (b) compute $B = g_1 \cdot h_0^s \prod_{i=1}^{\ell} h_i^{m_i}$, (c) compute $A = B^{1/(x+e)}$.
    Output signature $\sigma=(A,e,s)$.
    \item $\mathsf{Verify}(pk,\sigma,(m_1,\dots,m_\ell))$ checks: $e(A, pk\cdot g_2^{e}) \stackrel{?}{=} e\!\left(g_1 \cdot h_0^{s}\prod_{i=1}^{\ell}h_i^{m_i},\, g_2\right)$.
    
\end{enumerate}
Here, $g_1 \in \mathbb{G}_1$ and $g_2 \in \mathbb{G}_2$ are fixed generators, and $\{h_i\}$ denotes the per-attribute generators used to bind the message vector. To achieve verifiable privacy, we employ Schnorr-style ZK proofs and the Fiat–Shamir transformation~\cite{fiat1986prove}. For a statement of knowledge $\mathcal{R}(x,w)$, the prover demonstrates knowledge of a witness $w$ satisfying the relation without revealing $w$.  
For example, to prove knowledge of discrete log $x$ such that $y = g^x$, the prover chooses random $r$, computes $t = g^r$, and derives a challenge $c = H(g, y, t)$. The response $s = r - c x \bmod p$ is sent to the verifier, who checks: $g^s y^c \stackrel{?}{=} t$.

\section{System and Framework Overview}
\label{sec:overview}

\subsection{System Participants}

The \sysname~model introduces four primary actors:
\begin{itemize}
    \item \textbf{Type-A (Regulatory) Issuer:} Trusted credential issuer, such as a government authority, that certifies regulatory eligibility and oversees compliance but does not directly manage user payments. \issA~mints and regulates the token used for payment, and deploys the escrow smart contract that settles payments based on usage.
    \item \textbf{Type-B (Commercial) Issuer:} Credential issuers that issue commercial or economic credentials such as deposit amounts on escrow contract and user limits. MNOs or naval ships that handle multiple users are considered \issB. The payment to \emph{BM}s is settled via the escrow contract, so no direct communication between \emph{BM}(s) and \issB~is necessary. Since \issB~can obtain this token from~\issA~via a standard transaction, we omit further discussion of this transfer mechanism.

    \item \textbf{Band Manager (BM):} Authorized intermediary responsible for granting and monitoring usage, collecting heartbeats from active sessions, and claiming payments from the escrow contract based on verified proofs of service. User's identity is always hidden from \emph{BM}s.
    \item \textbf{User (Device):} The end user or network device that requests temporary spectrum access. Each user maintains unlinkable session credentials and interacts with the \emph{BM}~through verifiable proofs.
\end{itemize}

{\sysname~maps onto existing roles, so no new actor is required. The regulator (\issA) can mandate the privacy-preserving layer as a condition of SAS certification, giving it system-wide leverage; \emph{commercial issuers} and \emph{BM} adopt it for non-repudiable billing and dispute evidence while shedding metadata-storage liability; and users, especially public-sector incumbents, gain operational-privacy protection.}

\subsection{Technical Challenges and Design Intuition}
{Designing a PAYG-based spectrum-sharing system that simultaneously achieves privacy, accountability, and real-time enforcement poses several challenges: (i) \emph{binding without linking}---the hidden device identifier must be provably identical across both credentials, yet never become a cross-session correlation handle to ensure privacy; (ii) \emph{metering under anonymity}---usage must be billed over time without a linkable per-interval signature; (iii) \emph{trust-minimized settlement}---an untrusted \emph{BM} must be paid exactly for verified usage via a contract holding only constant-size state; and (iv) \emph{auditability}---a regulator must resolve disputes without the chain or \emph{BM} ever learning a real identity. 
}

\textbf{Privacy vs. auditability} is addressed by decoupling identity from settlement. Identity-bound credentials are handled off-chain using unlinkable BBS\(+\) selective-disclosure proofs, while the blockchain records only ephemeral session identifiers and hash-chain commitments for auditable usage. Fresh session IDs prevent cross-session linkability, while \issB~can maintain an off-chain mapping to real device identities for regulatory audit when needed.

\textbf{Fair settlement and collusion resistance} are enforced directly at the contract level. \issB’s deposit and the \emph{BM}’s claims are bound by smart-contract logic, and each claim must include a verifiable hash-chain heartbeat. Invalid, replayed, or inflated claims are automatically rejected, ensuring that each usage interval is billed exactly once without relying on trust.

\textbf{Scalability and real-time operation} are achieved by generating all cryptographic proofs off-chain and committing only constant-size values on-chain. This keeps contract execution lightweight and compatible with the tight latency requirements of dynamic spectrum sharing.

\subsection{Threat Model}
We assume the following trust and threat model:

\noindent 1. \emph{Issuers} are honest entities that issue verifiable credentials to the users. \issA~(e.g., FCC) is trusted by all other system participants. \issB s (e.g., MNOs) operate within the user trust boundary, as users (e.g., CBSD devices) are typically affiliated with them. However, they are not trusted by the \emph{BM}s; instead, their trust is enforced through an on-chain smart contract deployed by the \issA. Both of the issuers are trusted for keeping the user’s real identity within their trust boundaries.\\
2. \emph{BM}s are untrusted and can be malicious, and they might want to claim payments for the portion that the users have not used. This is prevented by smart contract bounded cryptographic proofs, where the contract validates the claims and pays the \emph{BM}s accordingly. \\
3. Users are untrusted by default; hence, all usage claims must be cryptographically verifiable through hash chained \hb s and session ID.\\
4. The blockchain and smart contracts are tamper-resistant under standard Byzantine fault tolerance assumptions~\cite{zhang2024reaching}.

\subsection{System Overview and Workflow}
As shown in Figure~\ref{fig:model_flow}, the high-level workflow of the system is as follows:

\begin{figure}[h]
    \centering
    \includegraphics[width=.7\textwidth]{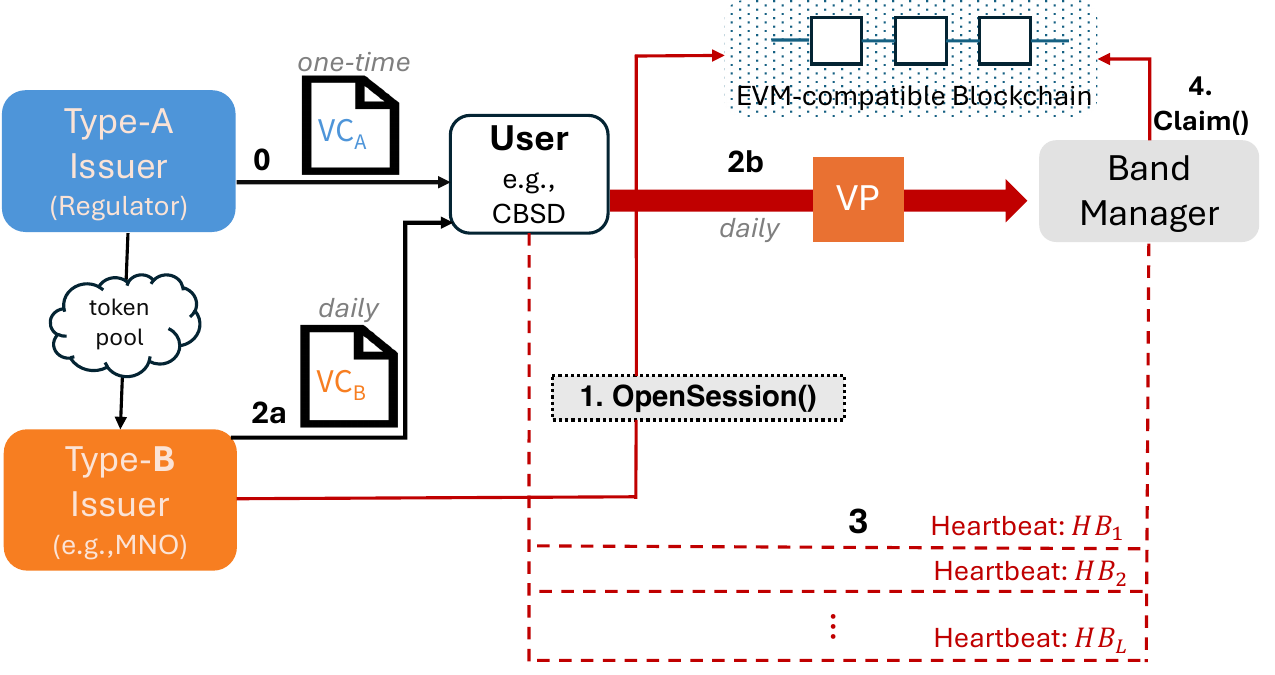}
    \vspace{-6pt}
    \caption{High-level Workflow of \sysname. Step 0 happens once, while Steps 1-4 take place on a daily basis.}
    \label{fig:model_flow}
    \vspace{8pt}
\end{figure}

\textbf{Step 0. Regulatory Credential Issuance:} Each user obtains a long-term regulatory credential $VC_A$ from \issA. This credential certifies the user’s license or eligibility using BBS$+$ signatures, enabling future selective disclosure of attributes without revealing the user’s identity. This issuance happens one time as long as there is no change in the regulatory terms or certifications.

\textbf{Step 1. Session Initialization:} Users are granted spectrum usage access on a daily basis, which we model as a \emph{session} in our system. \issB~initializes a session for the escrow smart contract by depositing tokens in the contract on behalf of the user. The contract records session parameters, such as, session ID ($s\_id$), unit price, limit, deposit amount, etc. The smart contract keeps track of the current session and updates the states according to \emph{BM}'s claim and payouts. Before initializing the session, \issB~sends the session ID to the user and gets the required information, such as the hash chain root $y_0$.

\textbf{Step 2a. Commercial Credential Issuance:} 
\issB~issues a commercial credential $VC_B$ to the user, also constructed using BBS$+$ signatures, valid for a specific session (e.g., daily) and cryptographically binding the user to the session parameters from Step 1. Session initialization precedes $VC_B$ issuance to ensure atomicity: only after issuance can the user present credentials and transmit heartbeats, and the band manager claim payments.

\textbf{Step 2b. Anonymous Spectrum Access Request:} The user proves the validity of both credentials ($VC_A$ and $VC_B$) by submitting a \emph{verifiable presentation} (VP) to \emph{BM}. This proof ensures that the device identity remains hidden and long-term history (across multiple sessions) is unlinkable. \emph{BM}~verifies the VP and grants user access.

\textbf{Step 3. Heartbeats Transmission and Verification:} During active transmission, the user periodically transmits \hb s (e.g., every 5 minutes) using precomputed hash chain values and session ID. Each heartbeat serves as a verifiable indicator of the ongoing channel occupancy. After receiving each \hb, the \emph{BM}~verifies it by generating the hash values for that specific index using previously received hashes.

\textbf{Step 4. Claiming and Automated Settlement:} After collecting a batch of \hb s, the BM then submits a claim request to the escrow contract, attaching the latest verified heartbeat along with the session ID. The \emph{BM}~can claim payout for any number of cumulative \hb s. The smart contract validates the claim against session metadata and releases proportional payment from the escrow to the \emph{BM}'s wallet. \issB~can close the session after the expiration time.

\section{\sysname Detailed Design}
\label{sec:system_design}

\subsection{System Setup}
\label{subsec:system_setup}
Before any PAYG interaction begins, a device or user must first determine whether spectrum is available in its geographical zone. In practice, this can be achieved by querying the BM or a public spectrum-availability index exposing only occupancy data. This discovery step is orthogonal to our protocol and independent of its cryptographic and contractual mechanisms. Accordingly, our system setup assumes that availability information has already been obtained, session ID is generated by \issB~and sent to the user. The PAYG workflow begins from that point onward. 

\textbf{Hash Chain Setup.}
To meter spectrum usage in discrete intervals, the user constructs a forward-secure hash chain of length $L$, where $L$ is the number ($limit$) of heartbeats the user can send within that session. The generation follows:

$\qquad \qquad \qquad y_L \overset{\$}{\leftarrow} \{0,1\}^\lambda,
\qquad 
y_{i-1} = H(s\_id \parallel i-1 \parallel y_i)$

Sequentially, the user calculates from $y_L$ to hash-chain root $y_0$. After generating the hash chain, the user sends hash chain root $y_0$ to \issB~when sending the request for credential $VC_B$ issuance. The user also mentions the limit $L$, along with other required parameters. These hash values of the chain work as the proof of cumulative \hb s. As all the $y_i$ are computed with $y_{i+1}$, it becomes impossible for the BM to pre-compute any $y_i$ to claim payment settlement for the \hb~that is not sent by the user.

Hash chains offer three key advantages: (1) \emph{low-overhead} authentication—each heartbeat transmits only a hash, eliminating costly signatures; (2) \emph{tamper evidence}—missing or out-of-order heartbeats are detectable by rewinding the chain using a prior verified hash and index; and (3) \emph{compact verification}—the band manager stores only the latest valid hash, with consistency against the initial commitment $y_0$ guaranteeing full sequence integrity. In \sysname, hash chains compactly encode cumulative spectrum usage, enabling the contract to verify aggregate cost from the latest hash in the band manager's proof.

\textbf{Credentials Issuance.}
Let $\mathbb{G}_1, \mathbb{G}_2$ be cyclic groups of prime order $p$ with generators $g_1 \in \mathbb{G}_1$, $g_2 \in \mathbb{G}_2$, and a bilinear map $\hat{e} : \mathbb{G}_1 \times \mathbb{G}_2 \rightarrow \mathbb{G}_T$.  
The issuer’s secret key is $x \leftarrow \mathbb{Z}_p$, and the public key is $w = g_2^x$.  
For attribute bases $h_0, h_1, \dots, h_\ell \in \mathbb{G}_1$, a message vector $m_1,\dots,m_\ell \in \mathbb{Z}_p$, and random $s \in \mathbb{Z}_p$, the issuer computes: $B = g_1 h_0^s \prod_{i=1}^{\ell} h_i^{m_i}, \quad
A = B^{1/(x + e)} \in \mathbb{G}_1,$ 
where $e \leftarrow \mathbb{Z}_p$ is a randomization exponent.  
The signature on the message set is $(A, e, s)$, and verification checks: $\hat{e}(A, w \cdot g_2^{e}) \stackrel{?}{=} \hat{e}(g_1 h_0^s \prod_{i=1}^{\ell} h_i^{m_i}, g_2).$ 

Each device is issued two sets of verifiable credentials, one from each authority, that share the user's device ID as a common and hidden attribute. The \emph{Type-A credential} $VC_A$, issued by \issA, contains regulatory attributes such as access tier (e.g., Priority Access License (PAL) or General Authorized Access (GAA)), device class, and regulatory expiration time, etc. The \emph{Type-B credential} $VC_B$, issued by \issB, contains commercial attributes, including pricing parameters, heartbeat, usage limit, unique session ID, and a commercial expiration time, etc.


\subsection{Smart Contract Session Initialization}

\begin{algorithm}[t]
\caption{\small Opening Escrow Session by Type-B Issuer}
\label{alg:open-session}
\scriptsize
\begin{algorithmic}[1]

\STATE \texttt{OpenSession}($\texttt{s\_id}, y_0, L, \texttt{unit\_price}, \texttt{deposit}, \texttt{expB}, bm\_addr$) \algcomment{session initialization}

\REQUIRE $!\texttt{Sessions}[s\_id].\texttt{open} \wedge deposit \ge unit\_price \wedge expB > \textsf{now}$ \algcomment{parameter verification}

\REQUIRE $\textsf{ERC20.transferFrom}(\textsf{issuerB} \rightarrow \textsf{contract}, \textsf{deposit}) = \textsf{true}$ \algcomment{token transfer to contract}

\STATE $\texttt{Sessions}[s\_id] \gets \{ \texttt{issuerB}=\textsf{caller},~ \texttt{bm}=bm\_addr, \texttt{deposit}=deposit,$
\STATE \hspace{2.5cm} $\texttt{spent}=0,\texttt{unit\_price}=unit\_price, \texttt{idx\_last}=0,$ 
\STATE \hspace{2.5cm} $\texttt{y\_last}=y_0, \texttt{L}=L,~ \texttt{expB}=expB, \texttt{open}=\textsf{true} \}$ \algcomment{set default values}

\STATE \textbf{emit} $\texttt{SessionOpened}(s\_id, \textsf{caller}, bm\_addr, L,expB, unit\_price, deposit)$

\end{algorithmic}
\end{algorithm}

Before sending $VC_B$ to the user, \issB~initiates a payment session for the escrow contract. The purpose of this is to verify that it matches the prerequisites (e.g., \issB~has more balance than $deposit$ amount). As seen in algorithm~\ref{alg:open-session}, \issB~initiates the session by sending: $s\_id$ (unique session ID), $y_0$ (hash-chain root), $L$ (limit on number of heartbeats), $unit\_price$ (per-heartbeat spectrum price), $deposit$ (escrow amount), $expB$ (expiration time of $VC_B$), and $bm\_addr$ (band manager's address) as parameters. Upon successful verification of the parameters, the contract initializes the session state with default values $\texttt{spent}=0$, $\texttt{idx\_last}=0$, and $\texttt{y\_last}=y_0$. The deposit is pulled from \issB's account and locked in escrow. Here, $idx\_last$ and $y\_last$ are the latest \hb~index counter and the latest hash value that the BM claimed for that session, respectively, and will be updated after each settlement. 

This phase caps \issB's financial exposure to the escrowed deposit and opens the session for BM settlement via authentic heartbeats. \issB~then issues $VC_B$ to the user, enabling credential presentation and immediate heartbeat transmission.

For auditability purpose, \issB~maintains a registry by mapping \texttt{s\_id} with \texttt{device\_id}. With this registry, they can later audit the usage of a specific user by retrieving the on-chain data, as \texttt{s\_id} is published on-chain.  

\subsection{Selective-Disclosed Verifiable Presentation (VP)}
After obtaining both credentials, the device proves to the BM that it holds valid BBS$+$ signatures $\sigma_A$ and $\sigma_B$ on $VC_A$ and $VC_B$, that all disclosed regulatory and commercial attributes are correctly signed, and that the hidden attribute $m_{id}$ is consistent across both credentials. All proofs are generated off-chain and verified by the BM. Concretely, the user computes randomized commitments \(T_{1_A}, T_{2_A}, T_{1_B}, T_{2_B}\). All commitments are bound together by a single Fiat--Shamir challenge \(c\), and the user returns the response vector \((e, r_1, r_2, s, m_{\mathsf{id}})\), forming a standard selective-disclosure proof of knowledge of both signatures and the hidden attribute.
Upon receiving the presentation, the BM deterministically recomputes all commitments, reconstructs the challenge using the same transcript, and verifies that the pairing equations for both signatures hold under the disclosed subset of attributes. Verification succeeds only if all algebraic relations, signature equations, and zero-knowledge responses are satisfied, ensuring that the presentation is valid, unforgeable, and unlinkable.
We follow the standard BBS$+$ protocol with ZKP for credential issuance and verification; full proof procedures are deferred to the appendix.

\subsection{Heartbeat Transmission}
In our model, a \hb~is a lightweight proof-of-progress sent periodically from the device to the BM to account for metered usage. One option is to sign each \hb~with the user's session key for the BM to verify, but at a fine interval (e.g., 5 minutes) per-\hb~signature verification is costly. Instead, we derive each \hb~from a forward-secure hash chain, generated as in Section~\ref{subsec:system_setup} and uniquely bound to the PAYG session. A transmitted heartbeat at index $i$ is $\mathrm{HB}(i) = (s\_id,\, i,\, y_i)$.

\noindent \textbf{Verification.}
Upon receiving $\mathrm{HB}(i)$, the BM reconstructs the chain value of the previous index by computing $\hat{y}_{i-1} = H(s\_id \,\|\, (i-1) \,\|\, y_i)$.
If $\hat{y}_{i-1}$ matches the last verified value $y_{i-1}$ stored in state, the heartbeat is accepted, ensuring freshness and replay resistance.

\noindent \textbf{Handling Out-of-Order or Missing Heartbeats.}
The BM can authenticate skipped indices by iteratively hashing backward from the received value. For example, if $\mathrm{HB}(5)$ arrives after $\mathrm{HB}(2)$, the BM computes:

$    y_4 = H(s\_id \,\|\, 4 \,\|\, y_5), \quad
    y_3 = H(s\_id \,\|\, 3 \,\|\, y_4), \quad
    y_2 = H(s\_id\,\|\, 2 \,\|\, y_3),$
    
until it regenerates the previously validated $y_2$. If the recomputed value matches the stored chain tip, all intermediate heartbeats are implicitly authenticated and can be credited accordingly.
This construction provides replay resistance, implicit aggregation, and forward integrity, enabling efficient metering with minimal communication overhead.

\subsection{Payment Settlement Using Heartbeats}
\begin{algorithm}[h]
\caption{\small Claiming Payment by \bm~using \hb}
\label{alg:settle-payment}
\scriptsize
\begin{algorithmic}[1]

\STATE \texttt{Claim($s\_id, i, y_i$)} \algcomment{payment claim call with $s\_{id}$ and hash chain value}

\REQUIRE $\textsf{caller} = \texttt{Sessions}[s\_id].\texttt{bm}$ \algcomment{checks if band manager is authorized}
\STATE $s \gets \texttt{Sessions}[s\_id]$
\STATE \textbf{Require} $s.\texttt{open} = \textsf{true}; \textsf{now} < s.\texttt{expB}; s.\texttt{idx\_last} < i \le s.\texttt{L}$ \algcomment{checks conditions}
\IF {\textbf{PreimageCheck}$(y_i, i, idx\_last, y\_last)=1$}
\STATE $\Delta \gets i - s.\texttt{idx\_last}, payout \gets \Delta \cdot s.\texttt{unit\_price}$
\STATE \textbf{Require} $s.\texttt{spent} + payout \le s.\texttt{deposit}$
\STATE $s.\texttt{idx\_last} \gets i, s.\texttt{spent} \gets s.\texttt{spent} + payout, s.y\_last \gets y_i$

\STATE \textbf{Require} $\textsc{ERC20.transfer}(s.\texttt{bm},~ payout) = \textsf{true}$

\STATE \textbf{emit} $\textsc{Claimed}(s\_id,~ i,~ \Delta,~ payout)$
\ENDIF

\end{algorithmic}
\end{algorithm}

During the session, the BM receives and then submits cumulative \hb s to obtain a proportional payment. The BM can claim settlement for any incremental \hb s. 
As shown in algorithm~\ref{alg:settle-payment}, to claim payment for \hb s $idx\_last$ through $i$, the BM calls $\textsf{Claim}(\textsf{s\_id}, i, y_i)$.

The contract checks the preconditions, and then checks the authenticity of the claimed $y_i$ by generating it using $y\_last$ and $idx\_last$. It recomputes the hash-chain links from $i$ down to the last claimed index $idx\_last$: $y_{j-1}' = H({s\_id}\parallel j-1 \parallel y_j'):\text{for } j = i,i-1,\dots,{idx\_last}+1$.

\begin{definition}[Valid Heartbeat Proof]
A submitted pair $(i,y_i)$ is \emph{valid} iff ${PreimageCheck}(y_i, i, {idx\_last}, {y\_last}) = 1,$ where, 
\[
{PreimageCheck} =
\begin{cases}
1, & \text{if } H({s\_id}\parallel j-1 \parallel y_j)=y_{j-1}; \forall j\in[{idx\_last}+1, i],\\
0, & \text{otherwise}.
\end{cases}
\]
\end{definition}
\noindent Only if the final computed $y'$ equals the stored $\texttt{y\_last}$ the claim will succeed.

\textbf{Payout Computation.} 
If verification succeeds, the number of newly revealed beats is: $\Delta = i - {idx\_last}.$ The contract computes the payment: ${payout}=\Delta \cdot {unit\_price}$. Then the states are updated: ${spent} \gets  {spent} + {payout}, {idx\_last} \gets i, {y\_last} \gets y_i$, and transfers the payout to the BM. The contract enforces three critical protections: (1) \emph{replay resistance}—each heartbeat can be claimed only once; (2) \emph{overcounting prevention}—the contract deterministically recomputes hash-chain links between the last claimed and newly submitted index, ruling out skipped, injected, or inflated intervals; and (3) \emph{session binding}—each $y_i$ is domain-separated to $(\textit{s\_id}, i)$, preventing cross-session reuse or adversarial mixing of hash-chain states.


\subsection{Session Closing}

\begin{algorithm}[h]
\caption{Closing Escrow Session by Type-B Issuer}
\label{alg:close-session}
\scriptsize
\setlength{\algorithmicindent}{1.2em}
\begin{algorithmic}[1]
  \STATE \texttt{CloseSession}($s\_id$) \hfill \textit{// session closing call}
  \REQUIRE $\textsf{caller} = \texttt{Sessions}[s\_id].\texttt{issuerB}$
  \STATE $s \gets \texttt{Sessions}[s\_id]$
  \REQUIRE $s.\texttt{open} = \textsf{true}, s.\texttt{expB}<\textsf{now}$
  \STATE $s.\texttt{open} \gets \textsf{false}$
  \STATE $refund \gets s.\texttt{deposit} - s.\texttt{spent}$ \hfill \textit{// remaining balance}
  \IF{$refund > 0$}
      \REQUIRE $\textsc{ERC20.transfer}(s.\texttt{issuerB},~refund) = \textsf{true}$
  \ENDIF
  \STATE \textbf{emit} $\textsc{Closed}(s\_id,~refund)$
\end{algorithmic}
\end{algorithm}

As shown in Algorithm~\ref{alg:close-session}, \issB~may close the session after the expiration of $VC_B$ ($expB$). $expB$ ensures the BM retains sufficient time after session end to claim outstanding payments, as claims remain valid until $expB$.
The contract refunds any unused deposit to \issB's account. 
Once the session is closed, no further heartbeats can be claimed, ensuring a well-defined settlement process, as any remaining escrow is deterministically refunded.

\section{Security Analysis}
\label{sec:security_analysis}

We consider four core security properties: unlinkability, unforgeability, fair settlement, and collusion resistance for our \sysname~system. Our guarantees rely on standard assumptions, including the q-SDH assumption in $(\mathbb{G}_1,\mathbb{G}_2)$~\cite{boneh2004short}, discrete logarithm hardness in the Pedersen commitment group $\mathbb{G}_T$, collision and preimage resistance of the hash function $H$, the random oracle model for Fiat-Shamir challenges~\cite{fiat1986prove}, and standard blockchain consensus properties~\cite{xiao2023bd,grissa2019trustsas}. For completeness, readers are referred to the appendix for detailed security proofs.

\textbf{Theorem 1 (Unforgeability).} \emph{If BBS${+}$ signatures are existentially unforgeable under chosen-message attacks, then no probabilistic polynomial-time adversary can forge a valid verifiable presentation in our construction.}

\textit{Proof sketch:}  
A valid forged presentation implies the ability to produce a valid BBS${+}$ signature $(A, e, s)$ without interacting with the issuer, which contradicts the unforgeability of BBS${+}$ signatures issued by both issuers.

\textbf{Theorem 2 (Unlinkability).}
\emph{
If the BBS${+}$ based verifiable presentation protocol is zero-knowledge
and all commitments used in the proof are perfectly or statistically
hiding, then for every probabilistic polynomial-time adversary
$\mathcal{A}$,
$
\left|
\Pr\big[ \mathsf{Expt}^{\text{unlink}}_{\mathcal{A}}(\lambda) = 1 \big]
- \frac{1}{2}
\right|
$
is negligible in the security parameter $\lambda$. In particular, the BM
cannot computationally distinguish whether two presentations are generated
from the same hidden identifier or from two different identifiers.
}

\textit{Proof sketch:}  
Each presentation includes freshly randomized BBS${+}$ commitments and Schnorr-style proof responses. By the zero-knowledge property, a simulator can produce transcripts indistinguishable from real presentations without access to the hidden identifier or any secret attributes. Thus, no distinguisher can link two transcripts without breaking the zero-knowledge of the proof system or the randomization of BBS${+}$ signatures.

\textbf{Theorem 3 (Fair Settlement).}
\emph{Fair settlement guarantees that payments reflect actual spectrum usage. If the hash function $H(\cdot)$ is preimage-resistant, then a malicious band manager cannot increase the settlement amount without providing a valid sequence of heartbeats. On the other hand, a malicious user cannot receive service without producing payment-generating heartbeats.
}

\textit{Proof sketch:}  
To claim more beats than transmitted, an adversary would need to forge $y_i$ without knowledge of $y_{i+1}$, which requires inverting the hash function. Conversely, without providing payment-generating heartbeats with valid presentations, a user cannot receive valid authorizations to use the spectrum. 

\textbf{Theorem 4 (Collusion Resistance).}
\emph{
No coalition of malicious band managers can link two spectrum access requests from the same user or manipulate the settlement outcomes.
}

\textit{Proof sketch:}  
Zero-knowledge of hidden attributes prevents leakage across presentations. Settlement logic is deterministic and depends only on valid heartbeats, preventing manipulation.

We additionally verify \sysname~symbolically in Tamarin~\footnote{https://tamarin-prover.com/}, confirming: (i) sessions open only after valid issuance-backed presentations, (ii) accepted claims are linked to their session, (iii) accepted claims match previously generated values, and (iv) no claim is accepted twice. This complements the formal proofs: the theorems establish primitive security, while Tamarin rules out protocol-logic flaws such as replay or inconsistent settlement.

\section{Evaluation}

\label{sec:evaluation}

This section provides the complexity analysis and evaluates the computational performance, scalability, and on-chain costs of our proposed \sysname protocol. We demonstrate that \sysname~authorization is practical for real-world deployment and that on-chain settlement supports high-frequency micro-payments. We evaluate two dimensions: (i) cryptographic computational costs; and (ii) gas consumption and storage overhead of on-chain operations, and also discuss confirmation latency for payment settlement. All experiments were executed on a Linux-based machine with a 12th Gen Intel(R) Core(TM) i7-12700K, 32GB RAM, and NVIDIA GeForce RTX 3070Ti GPU, unless otherwise mentioned.

\begin{table}[t]
\centering
\caption{Asymptotic Complexity of Cryptographic and Contract Operations}
\label{tab:complexity}

\resizebox{0.65\textwidth}{!}{%
\scriptsize
\begin{tabular}{lcc}
\toprule
\textbf{Operation} & \textbf{Cost} & \textbf{Dominant Primitives} \\
\midrule

$VC_A$ issuance &
$O(k)$ &
$k$ multiexp + 1 pairing \\

$VC_B$ issuance &
$O(k')$ &
$k'$ multiexp + 1 pairing \\

Signature randomization &
$O(1)$ &
Constant multiexp \\

Blinded triple $(\tilde A, D, \tilde B)$ &
$O(1)$ &
3 multiexp \\

Schnorr \#1 commitment &
$O(1)$ &
3 exponentiations \\

Schnorr \#2 commitment &
$O(1)$ &
1 exponentiation \\

Fiat-Shamir challenge &
$O(k+k')$ &
Hash over $n$ transcript items \\

Response computation &
$O(1)$ &
Constant exponentiations \\

Verifier recomputation &
$O(k + k')$ &
$(k + k')$ multiexp + 2 pairings \\

Verify FS challenge &
$O(k+k')$ &
Hashing \\

Hash chain setup &
$O(L)$ & 
Hashing \\

Heartbeat verification &
$O(1)$ &
1 hash \\

\texttt{openSession}() &
$O(1)$ &
SSTORE + ERC20 Transfer \\

\texttt{claim}() &
$O(\Delta)$ &
$\Delta$ hashes + SLOAD + SSTORE \\

\texttt{closeSession}() &
$O(1)$ &
SSTORE + ERC20 Transfer \\
\bottomrule
\end{tabular}
}
\vspace{8pt}
\end{table}

\subsection{Complexity Analysis}
Table~\ref{tab:complexity} summarizes the asymptotic costs of cryptographic and contract-level operations. $VC$ issuance requires linear-time multi-exponentiation in $\mathbb{G}_1$ and one pairing, scaling as $O(k)$ and $O(k')$ for $VC_A$ and $VC_B$ ($k{=}4$, $k'{=}9$ in our instantiation). Credential randomization and blinded-triple $(\tilde{A},D,\tilde{B})$ computation cost $O(1)$ exponentiations, preserving unlinkability; Schnorr commitments and Fiat--Shamir challenge generation are likewise constant given a fixed-size transcript. Verification is dominated by $O(k+k')$ group operations and two pairings. Chain initialization costs $O(L)$ hashing, while per-heartbeat emission and verification are $O(1)$. \texttt{OpenSession}() and \texttt{CloseSession}() are constant-time and \texttt{Claim}() is $O(\Delta)$. Overall, cryptographic costs scale linearly with credential size; all else is constant. Keeping expensive operations off-chain yields low device overhead, efficient verification, and predictable on-chain execution suited to high-frequency metering.

\subsection{Cryptographic Implementation and Results}
All measurements in Table~\ref{tab:crypto_summary} were obtained using a Python benchmarking harness that wraps Dock Network’s Rust implementation of BBS$+$ over the BLS12-381 curve via PyO3 bindings. The harness was implemented in Python~3.10.12, invoking Rust code compiled with \texttt{rustc~1.91.0} in release mode, and uses Dock’s \texttt{bbs\_plus} and \texttt{proof\_system} crates. This setup reflects optimized BLS12-381 arithmetic without debug overhead. All reported averages are computed over 1000 iterations of each step.

\begin{wraptable}{r}{0.42\textwidth}
  \centering
  \caption{Average computation time for cryptographic steps}
  \label{tab:crypto_summary}
  \vspace{-0.5 em}
  \scalebox{0.8}{
  \begin{tabular}{lc}
    \toprule
    Step & Mean (ms) \\
    \midrule
    Issuance $VC_A$ & 10.3 \\
    Issuance $VC_B$ & 12.2 \\
    Proof generation & 45.9 \\
    Proof generation (RPi) & 189 \\
    Proof verification & 84.6 \\
    Hash chain  & 0.14 \\
    Hash chain (RPi) & 0.59 \\
    Proof size & $\approx$857 bytes \\
    \bottomrule
  \end{tabular}
  }
\end{wraptable}

\emph{$VC_A$ issuance} yields a 10.3~ms average. \emph{Issuance $VC_B$} signs a larger attribute vector and therefore incurs a slightly higher cost (12.2~ms). \emph{Proof generation} by the user constructs a selective-disclosure proof, generating Schnorr commitments, and proving equality of the hidden device identifier, resulting in a 45.9~ms average. \emph{Proof verification} by the BM recomputes the randomized BBS$+$ checks and witness-equality constraints using multi-exponentiations and pairings, with an average cost of 84.6~ms for a proof of size $\approx$857~bytes. \emph{Hash-chain generation} captures the one-time cost of building the heartbeat chain with SHA-256. We model a session as 24 hours; a 5-minute heartbeat interval yields 288 heartbeats, so we set the limit $L = 300$, requiring only 0.14~ms. Each \emph{heartbeat emission} is $O(1)$, below 0.1~ms. To assess practicality, we also measured the main user-side operations on a Raspberry Pi 5. {We use a Raspberry Pi 5 as a proxy for a real CBSD's host-compute envelope: deployed CBRS small cells run embedded Linux on ARM Cortex-A application processors~\cite{marvell-armada,cloudran-arm}, the same ARMv8-A family as the Pi 5's Cortex-A76~\cite{arm-a76-trm}.} There, proof generation averages 189~ms and hash-chain generation 0.59~ms, confirming that all off-chain operations are lightweight and practical on edge-class hardware, compatible with real-time session establishment and high-frequency metering.

\subsection{On-chain Gas Consumption}

We benchmark the \texttt{Escrow} contract's deployment and runtime gas on an EVM-compatible blockchain using Foundry with Anvil, which provides a deterministic, reproducible local backend free from network variability. Each experiment deploys a fresh contract and runs the full workflow---\texttt{OpenSession}, \texttt{Claim}, \texttt{CloseSession}---varying the heartbeat gap $\Delta$ per claim (1, 10, 50, 100 beats) to evaluate hash-chain verification costs.

\begin{wraptable}{r}{0.52\textwidth}
  \vspace{-1.3em}
  \centering
  \caption{On-chain Gas Consumption}
  \label{tab:escrow-gas-anvil}
  \scalebox{0.8}{
  \begin{tabular}{l|c|c|c|c}
    \toprule
    Function & Min & Avg & Median & Max \\
    \midrule
    OpenSession  & 251{,}762 & 251{,}767 & 251{,}762 & 251{,}774 \\
    Claim        & 31{,}794  & 122{,}254 & 109{,}719 & 304{,}587 \\
    CloseSession & 26{,}604  & 46{,}742  & 50{,}573  & 63{,}051 \\
    \bottomrule
  \end{tabular}
  }
\end{wraptable}
\noindent Table~\ref{tab:escrow-gas-anvil} reports deployment and per-function gas under Anvil defaults, with min/avg/median/max summarizing repeated executions. One-time deployment by \issA~costs $\approx$1.62M gas ($\approx$\$3 at the 60-day average Ethereum gas price~\footnote{https://etherscan.io/gastracker}), yielding a 7.3~kB image well below the EIP-170 limit. \texttt{OpenSession} incurs a fixed 251K gas, dominated by storage initialization and an ERC-20 \texttt{transferFrom}. \texttt{Claim} scales as $O(\Delta)$: $\Delta \leq 3$ costs 31K–50K gas, while 50–100 beats cost 0.11M–0.30M gas; since the band manager bears this cost, larger $\Delta$ lowers per-session overhead. \texttt{CloseSession} is inexpensive (26K–63K gas: constant-time arithmetic, a refund transfer, state updates). End-to-end average gas per session is $\approx$603K (at $\Delta{=}300$, $L{=}300$), or $\approx$\$0.9---negligible given the fairness and auditability guarantees, though it should be read as a representative aggregate under the measured workload, not a universal per-session constant. Post-closure, each session retains $\approx$288 bytes ($\approx$9 slots). Overall, \sysname~keeps on-chain costs bounded and predictable: constant setup, claims scaling with heartbeat gaps, and minimal storage.


\subsection{On-chain Settlement Latency}
\sysname's settlement logic is fully atomic: each \texttt{Claim} and \texttt{CloseSession} transaction verifies the hash-chain gap, updates session state, and transfers tokens within a single EVM execution. Settlement latency thus depends solely on block time and confirmation policy, not on multi-round interactions. On Ethereum L1, transactions are typically included within one block ($\approx$12~s), with strong economic finality after $\approx$12 confirmations (2-3~minutes). On rollup-based L2s (e.g., Optimism~\footnote{https://docs.optimism.io/} or Arbitrum~\footnote{https://arbitrum.io/rollup}), sequencer confirmation reduces this to 1-2~seconds, though full security requires awaiting finalized L1 inclusion; {thus we recommend confirming L1 finality before treating a \texttt{Claim} as final}.

\section{Conclusion}
\label{sec:conclusion}
This paper presents \sysname, a PAYG framework for dynamic spectrum sharing that enables verifiable, privacy-preserving micro-payments between spectrum users and band managers. By synergizing selectively-disclosed verifiable credentials with smart contract–based atomic settlement, \sysname tightly couples access authorization, fine-grained usage accounting, and auditable payment enforcement. Unlike general-purpose anonymous payment or credential systems, our design addresses the unique requirements of regulated, usage-metered services by tightly integrating privacy-preserving authentication, verifiable usage accounting, and programmable settlement. This domain-specific integration is essential for practical deployment in DSS environments and cannot be achieved through direct reuse of existing privacy-preserving solutions. Complexity analysis and an end-to-end implementation demonstrate that both cryptographic and on-chain costs remain lightweight. Furthermore, the implementation demonstrates that \sysname is practical for real-world dynamic spectrum sharing deployments, adding approximately $150\,\mathrm{ms}$ end-to-end latency alongside a modest on-chain gas cost of approximately $\approx 603K$ ($\approx \$0.9$) per session.
\\

\noindent \textbf{Acknowledgments.} This work was supported in part by the National Science Foundation under grants 2433904, 2433905, 2331936, 2332675, 2154929, 2247560, 2247561, 2442382, by the Office of Naval Research under grant N00014-24-1-2730, and by the Virginia Commonwealth Cyber Initiative (CCI).

\bibliographystyle{splncs04}
\bibliography{reference}





\appendix

\section{Appendices}

\subsection{Security Proofs}

\begin{proof}[\textbf{Unforgeability}]
Let $\mathcal{A}$ be any PPT adversary that outputs a forged presentation $\pi^\ast$ accepted by the verifier. \\ 

\textit{Event $\mathsf{E}_{\text{sig}}$: Forged BBS${+}$ Signature.}
Let $\pi^\ast$ contain a valid BBS${+}$ signature $(A^\ast,e^\ast,s^\ast)$ that was never issued. By standard extraction from the selective-disclosure $\Sigma$-protocol, the corresponding message vector $m^\ast$ can be recovered. We construct a reduction $\mathcal{B}$ that answers all issuance queries using the BBS${+}$ signing oracle and outputs $(m^\ast,(A^\ast,e^\ast,s^\ast))$ as a forgery. This contradicts BBS${+}$ EUF-CMA security, implying $\Pr[\mathsf{E}_{\text{sig}}] \le \operatorname{negl}(\lambda)$ by Lemma~\ref{lem:bbs-unforgeable}. Thus, the protocol is unforgeable.

\begin{lemma}[BBS$+$ EUF-CMA Unforgeability]
\label{lem:bbs-unforgeable}
    No PPT adversary can produce a valid BBS${+}$ signature on any message that was not previously queried to the signing oracle, except with negligible probability in $\lambda$.

\end{lemma}

\end{proof}

\begin{proof}[\textbf{Unlinkability}]
We prove unlinkability via a hybrid argument over the ZK simulator. $\mathsf{G}_0$ is the real experiment, with challenge presentations from real witnesses including the hidden identifier. $\mathsf{G}_1$ replaces all challenge proofs with simulated transcripts for the same statements; by Lemma~\ref{lem:zk-indistinguishability}, $\mathsf{G}_0\approx\mathsf{G}_1$. $\mathsf{G}_2$ swaps the hidden identifier with statements unchanged; as the simulator is witness-independent, $\mathsf{G}_1$ and $\mathsf{G}_2$ are identically distributed. By Lemma~\ref{lem:unlinkability}, presentations from different identifiers are indistinguishable, so the adversary's advantage is negligible:
$
\text{Adv}_{\mathsf{unlink}}^{\mathcal{A}}(\lambda)
=
\left|
\Pr[\mathsf{G}_0=1] - \tfrac{1}{2}
\right|
\le \operatorname{negl}(\lambda).
$
\end{proof}

\begin{lemma}[ZK Indistinguishability]
\label{lem:zk-indistinguishability}
Let $\pi$ be a real selective-disclosure proof generated with witness $w$,
and let $\pi_{\mathsf{sim}}$ be a simulated transcript for the same public
statement. Then for any PPT adversary, 
$\left|
\Pr[\mathcal{A}(\pi)=1]-\Pr[\mathcal{A}(\pi_{\mathsf{sim}})=1]
\right|
\leq
\operatorname{negl}(\lambda).
$

\end{lemma}

\begin{lemma}[Unlinkability of Presentations]
\label{lem:unlinkability}
For any two hidden identifiers $m_{\mathrm{id}}^{(0)}, m_{\mathrm{id}}^{(1)}$
and any PPT adversary with oracle access,
the distributions of their corresponding presentation transcripts are
computationally indistinguishable. That is, $\left|
\Pr[\mathsf{Expt}^{\mathsf{unlink}}_{\mathcal{A}}(\lambda)=1]
-
\frac{1}{2}
\right|
\leq
\operatorname{negl}(\lambda)$.
\end{lemma}


\begin{proof}[\textbf{Fair Settlement}]
The contract stores the last accepted heartbeat. A fairness violation requires accepting $(j,y_j)$ off the valid chain. We reduce this to a preimage attack: embed a challenge hash $z$ as $y_{j-1}$ and define $H'(x)=H(s\_{id}\parallel\!(j-1)\parallel x)$. By Lemmas~\ref{lem:prefix-preimage} and~\ref{lem:hashchain-soundness}, producing a consistent off-chain element is infeasible, so an inconsistent heartbeat is accepted only with negligible probability.

\begin{lemma}[Preimage Resistance]
\label{lem:prefix-preimage}
Let $H:\{0,1\}^*\rightarrow\{0,1\}^\lambda$ be a preimage-resistant hash function. Then for any fixed prefix string $\rho$, the derived function $H_\rho(x) = H(\rho \,\|\, x)$ is also preimage resistant against all PPT adversaries.
\end{lemma}

\begin{lemma}[Hash-Chain Soundness]
\label{lem:hashchain-soundness}
Given a hash chain defined by
$y_{i-1} = H(s\_{\mathrm{id}} \,\|\, (i-1) \,\|\, y_i)$,
no polynomial-time adversary can produce a value $y_j$
that is consistent with a previously accepted $y_{j-1}$ unless it knows
a valid successor $y_{j+1}$, except with negligible probability in~$\lambda$.
\end{lemma}

\end{proof}



\begin{proof}[\textbf{Collusion Resistance}]
Colluders reveal only their own credentials, attributes, and state, never honest users' hidden identifiers. By Lemma~\ref{lem:collusion-view}, this leaks no honest user's hidden attributes, so the adversarial view stays indistinguishable from the unlinkability experiment. If colluders could link two honest presentations with non-negligible advantage, a reduction simulating their secrets would break Lemma~\ref{lem:unlinkability}. Hence collusion links only with negligible probability.
\end{proof}

\begin{lemma}[View Preservation Under Collusion]
\label{lem:collusion-view}
    Colluding users learn no additional information about honest users' hidden attributes or witnesses beyond what is revealed in public presentations. Hence their joint view is computationally indistinguishable from the adversarial view in the unlinkability  experiment.
\end{lemma}

\subsection{Verifiable Presentation with Selective Disclosure}

\textbf{Blinded Triple Constructions.}
For each credential, the user samples: \\
$r_1, r_2 \xleftarrow{\$} \mathbb{Z}_q,
r_3 = r_2^{-1}\!\!\!\!\!\pmod{q}$.
Given the signer’s BBS$+$ signature $(A,e,s)$ and the attribute commitment $B = g_1 \cdot h_0^s \prod_i h_i^{m_i} h_{id}^{m_{id}}$, the user computes: \\
$\tilde{A} = A^{r_1 r_2},
D = B^{r_2},
\tilde{B} = D^{r_1} \cdot (\tilde{A})^{-e}.$
This blinded triple $(\tilde{A},D,\tilde{B})$ removes all linkability to the original signature while preserving the necessary algebraic relations for later pairing checks.

\textbf{Commitments for $VC_A$.}
The prover samples fresh randomness: $\alpha_{e,A}, \alpha_{r1,A},$ $\alpha_{r3,A},\alpha_{s,A}, \alpha_{id}
~\xleftarrow{\$}~\mathbb{Z}_q,$
and constructs (similar for $VC_B$):
$T_{1_A}= (\tilde{A}_A)^{\alpha_{e,A}} (D_A)^{\alpha_{r1,A}},$
$T_{2_A}= (D_A)^{\alpha_{r3,A}}\, h_{id}^{\alpha_{\mathrm{id}}}\, h_0^{\alpha_{s,A}}.$
Note that the same $\alpha_{\mathrm{id}}$ is used for both $VC_A$ and $VC_B$, which enforces that both credentials contain the same hidden device identifier.

\textbf{Fiat--Shamir Challenge.}
The challenge is generated over all transcript elements:
{\scriptsize
\[
\begin{aligned}
c &= \!\Big(
\texttt{CTX} \parallel 
W_A,W_B \parallel 
\tilde{A}_A,\tilde{A}_B
\parallel \text{reveal}_A,\text{reveal}_B \parallel 
(T_{1_A},T_{2_A},T_{1_B},T_{2_B})
\parallel s\_id \parallel \text{bm\_addr}
\Big)
\end{aligned}
\]
}
\textbf{Response Computation.}
The user computes $\widehat{m}_{\mathrm{id}} = \alpha_{\mathrm{id}} + c\, m_{id}$, and: 
{\scriptsize
\[
\begin{aligned}
\widehat{e}_A &= \alpha_{e,A} + c e_A,
&
\widehat{e}_B &= \alpha_{e,B} + c e_B,
& 
\widehat{s}_A &= \alpha_{s,A} + c s_A,
& 
\widehat{s}_B &= \alpha_{s,B} + c s_B, \\ 
\widehat{r}_{1_A} &= \alpha_{r1,A} - c r_{1_A},
&
\widehat{r}_{1_B} &= \alpha_{r1,B} - c r_{1_B}, 
& 
\widehat{r}_{3_A} &= \alpha_{r3,A} - c r_{3_A},
&
\widehat{r}_{3_B} &= \alpha_{r3,B} - c r_{3_B}.  
\end{aligned}
\]
}
\textbf{Verifiable Presentation (VP).}
The VP includes:
{\scriptsize
\[
\Big(
\text{reveal}_A,\text{reveal}_B, 
(\tilde{A}_A,\tilde{B}_A,D_A,T_{1_A},T_{2_A}), 
(\tilde{A}_B,\tilde{B}_B,D_B,T_{1_B},T_{2_B}), c, 
\widehat{e}_A,\widehat{r}_{1_A},\widehat{r}_{3_A},\widehat{s}_A, 
\]
\[
\widehat{e}_B,\widehat{r}_{1_B},\widehat{r}_{3_B},\widehat{s}_B,\ 
\widehat{m}_{id},\ 
s\_id, {bm\_addr}
\Big)
\]
}
This presentation is verified as follows: Schnorr proof~\#1 validates the blinded structure, Schnorr proof ~\#2 verifies the disclosed attributes and randomness, and a pairing check confirms the correctness of the underlying signature.

\textbf{Challenge Recompute.}
\bm~reconstructs the challenge $c'$ using the identical transcripts. Then it computes $B_{rev}$ from revealed attributes: $B_{rev} = g_1\prod_{i\in \text{rev}} h_i^{m_i}.$ For each $VC$, $B$ satisfies: $B = B_{\mathrm{rev}}\cdot h_{id}^{m_{id}}h_0^{s}$.

\textbf{Schnorr \#1: Blinded Structure Verification.}
BM checks, for $VC_A$: $\tilde T_{1_A} 
= (\tilde B_A)^{c} (\tilde A_A)^{\widehat e_A} (D_A)^{\widehat r_{1_A}}
\stackrel{?}{=} T_{1_A}.$

$VC_B$ follows identically.

\textbf{Schnorr \#2: Attribute and Randomness Verification.}
For $VC_A$, BM checks: $\tilde T_{2_A}
= (B_{\mathrm{rev},A})^{c} (D_A)^{\widehat r_{3_A}}
\,h_{id}^{\widehat m_{id}}\,h_0^{\widehat s_A}
\stackrel{?}{=} 
T_{2_A}.$

$VC_B$ follows identically.

\textbf{Pairing-Based Signature Verification.}
BM verifies that the blinded triple corresponds to a valid BBS$+$ signature. For $VC_A$: $e({\tilde A_A}, {W_A})
\stackrel{?}{=}
e({\tilde B_A}, {g_2}).$

This check ensures that a valid BBS$+$ signature under $W$ exists on the committed attributes. $VC_B$ follows identically.

\end{document}